\documentclass[letterpaper, 10 pt, conference]{ieeeconf}  % Comment this line out if you need a4paper

\IEEEoverridecommandlockouts                              % This command is only needed if 
\usepackage{graphics} % for pdf, bitmapped graphics files
\usepackage{epsfig} % for postscript graphics files
\usepackage{mathptmx} % assumes new font selection scheme installed
\usepackage{times} % assumes new font selection scheme installed
\usepackage{amsmath} % assumes amsmath package installed
\usepackage{amssymb}  % assumes amsmath package installed

\usepackage{arydshln}
\usepackage{balance}

\usepackage{float}
\usepackage{color}

\usepackage{algorithm}
\usepackage{algorithmic}

\newtheorem{remark}{\textbf{Remark}}
\newtheorem{theorem}{\textbf{Theorem}}
\newtheorem{lemma}{\textbf{Lemma}}
\newtheorem{ass}{\textbf{Assumption}}

\title{\LARGE \bf
Constrained finite-time stabilization by model predictive control: an infinite control horizon framework
}

\author{Bing Zhu, \emph{Senior Member, IEEE}, Xiaozhuoer Yuan,
	Zewei Zheng, \emph{Member, IEEE},\\
 and Zongyu Zuo, \emph{Senior Member, IEEE}% <-this % stops a space
\thanks{This work was supported by National Natural Science Foundation of China under grant 62073015.}% <-this % stops a space
\thanks{The authors are with The Seventh Research Division, Beihang University, Beijing 100191, P.R.China.}
\thanks{Emails: zhubing@buaa.edu.cn (B.Zhu); yuanxiaozhuoer@buaa.edu.cn (X.Yuan); zeweizheng@buaa.edu.cn (Z.Zheng); zzybobby@buaa.edu.cn (Z.Zuo)}
}
\begin{document}

\maketitle
\thispagestyle{plain}
\pagestyle{plain}

%%%%%%%%%%%%%%%%%%%%%%%%%%%%%%%%%%%%%%%%%%%%%%%%%%%%%%%%%%%%%%%%%%%%%%%%%%%%%%%%
\begin{abstract}
Existing results on finite-time model predictive control (MPC) often rely on terminal equality constraint, switching inside one-step region, or short control horizon, leading to limited initial feasibility.
This paper proposes an infinite-horizon Model Predictive Control framework for the constrained finite-time stabilization of discrete-time systems, overcoming limitations found in existing finite-time MPC results.
The proposed framework is built upon a terminal cost strategy, but expands it by replacing the short-horizon terminal cost with the sum of stage costs over an infinite control horizon. 
The infinite-horizon finite-time MPC is implementable via a finite-horizon MPC with a terminal cost. 
The proposed finite-time MPC is systematically applicable to constrained multi-input linear systems and a class of nonlinear systems.
\end{abstract}

%\begin{keywords}                           % Five to ten keywords,  
%Model predictive control; finite-time control; infinite-horizon optimization
%\end{keywords}

%%%%%%%%%%%%%%%%%%%%%%%%%%%%%%%%%%%%%%%%%%%%%%%%%%%%%%%%%%%%%%%%%%%%%%%%%%%%%%%%
\section{Introduction}

For discrete-time systems, finite-time stabilization %(or deadbeat control) 
means the closed-loop state converges exactly to zero in a finite number of time instants. 
Finite-time stabilization is possible only if the system is controllable \cite{Kalman1960On}.
A typical finite-time stabilization method for linear systems is to place all eigenvalues at the origin \cite{o1981discrete, kucera1984deadbeat}.
%Pioneering results on finite-time stabilization are provided in \cite{o1981discrete} and \cite{kucera1984deadbeat} and references therein.
Finite-time stabilization or control is useful for systems needing high accuracy and fast responses.
Potential applications include adaptive estimation for uncertainties in MPC \cite{zhu2020constrained}, deadbeat observers \cite{Ansari2019Deadbeat}, learning-based MPC \cite{huang2023LSTM}.
Recent results in finite-time stabilization and control can be found in \cite{haddad2020finite} and \cite{tuna2012state} and references therein for nonlinear systems.
%It is noted that 
%State or control constraints are rarely considered in the above references.

%It is noted that, for discrete-time systems, finite-time control aims to drive the state to the origin within a prescribed number of steps, adaptive and robust MPC schemes (e.g., \cite{zhu2020constrained}) focus on maintaining constraint satisfaction and stability in the presence of parametric uncertainties. These approaches typically guarantee asymptotic convergence through online parameter estimation or robust constraint tightening, rather than enforcing strict finite-time convergence. Thus, adaptive/robust MPC and finite-time control address fundamentally different objectives: robustness against uncertainty versus guaranteed convergence within finite time.

In principle, finite-time stabilization is often achieved through high-gain control to ensure a fast convergence rate around the origin, % \cite{Li2025Semi}, 
leading to aggressive closed-loop behavior. 
%However, practical systems are always subject to constraints, and aggressive closed-loop behavior possibly exceeds constraints, degrading or even destabilizing the actual system performance.
However, %all practical systems have constraints. 
aggressive closed-loop behavior can easily violate constraints, leading to performance degradation or even instability.
A typical strategy for unifying constrained control and finite-time stabilization is the formulation of model predictive control (MPC) schemes with finite-time convergence properties.
Some theoretical results can be found, e.g., in \cite{sutherland2019closed, 2018Sparsity, cheng2001stability}. 
%Applications of MPC with finite-time convergence can be found in \cite{kang2019symmetrical, wang2019improved, kreiss2021optimal}.
Applications of MPC with finite-time convergence are reported in interior permanent magnet synchronous motors \cite{kang2019symmetrical}, parallel interconnection of buck converters \cite{kreiss2021optimal}, finite-time trajectory tracking control of redundant manipulators \cite{jin2022finite}, finite-time control of mobile manipulators with floating-base \cite{su2025finite}, finite-time stabilization of DC electrical power systems in More electric aircraft \cite{zhang2024finite}, finite-time MPC for Permanent Magnet Synchronous Motor \cite{xie2015finite}, and MPC with AC deadbeat control for Modular Multilevel Converter-High-Voltage DC \cite{majstorovic2025deadbeat}.
In all these applications, finite-time MPC plays crucial role to achieve the desired performance and constraint satisfaction.

In MPC design, enforcing finite-time convergence typically requires prioritizing rapid convergence of states, which can compromise transient performance and exacerbate constraint violations.
As a result, several constrained finite-time MPC formulations have been developed to balance feasibility and convergence guarantees.
%
%Sub-optimal MPC in \cite{scokaert1999suboptimal} ensures the closed-loop trajectory enters a terminal invariant set before a finite time. If the terminal set shrinks to be the origin (which actually becomes a terminal equality constraint), then finite-time convergence performance can be expected. 
One approach is the suboptimal MPC framework in \cite{scokaert1999suboptimal}, which ensures entry into a terminal invariant set within finite time; when this set collapses to the origin, the method effectively becomes a terminal equality–based scheme with limited feasibility.
%The finite-time convergence performance can be achieved via 
%In ``one-step region'' strategy \cite{anderson2018finite},
%if the state enters one-step region in finite-time, it can be steered to the origin in one step without penalizing the control input. 
%However, in this type of MPC formulation, it requires penalizing the distance between the state and the one-step region, leading to extra switch strategies in implementation.
%This MPC formulation requires penalizing the distance between the state and the one-step region. This necessity, in turn, leads to extra switching strategies during implementation.
To avoid restrictive terminal conditions, the one-step–region strategy in \cite{anderson2018finite} steers the state into a region from which it can reach the origin in a single step, but this requires penalizing the distance to the region and introduces additional switching logic in implementation.
%
%Controller matching \cite{di2009model} can be applied to design finite-time MPC. However, it requires to solve optimization with a matrix equality constraint, which is possibly infeasible.
Controller-matching \cite{di2009model} provides an alternative way for achieving finite-time behavior by embedding a desired closed-loop controller into the MPC formulation, yet this requires solving an optimization problem with a matrix equality constraint, possibly leading to infeasibility in practice.
Practical finite-time MPC was provided in \cite{dreke2023practical}, where the matrix equality constraint is relaxed to a matrix inequality constraint, and approximate finite-time convergence performance is obtained.
In \cite{zhu2024finite}, a finite-time MPC was proposed, where the control horizon equals the system dimension, and only the terminal cost is penalized. While a terminal equality constraint ensures convergence, this approach suffers from a limited initial feasible region, especially for low-dimensional systems. 
Robust MPC variants seek to extinguish transients within a short horizon under disturbances \cite{schildbach2025deadbeat}. 
%Extensions to stochastic settings \cite{he2022finite} further show the diversity of strategies for achieving fast convergence.

Recent studies have further investigated finite-time convergence and feasibility issues within the MPC framework.
in \cite{comelli2023simplified}, a simplified MPC is proposed to replace the conventional terminal invariant set with simpler inner and outer sets. A terminal-set feasibility governor is developed in \cite{skibik2022terminal} to exploit the MPC terminal set for constraint satisfaction and finite-time convergence. Finite-time MPC has been extended to economic control problems to address economic dispatch and frequency regulation \cite{jia2023finite}. Finite-time convergence and input-to-state stability have been studied for MPC with disturbance preview information \cite{leung2024feasibility}.
However, a general finite-time MPC framework that can improve the initial feasible region while avoiding terminal equality constraints and additional switching strategies is still desirable.

In this paper, we propose a general finite-time MPC framework with infinite control horizon.
A key motivation for adopting this formulation, targeting finite-time convergence, is to overcome the limited initial feasibility inherent in short-horizon or terminal-equality–based finite-time MPC schemes. 
The distinguished feature of the proposed finite-time MPC is that stage costs are summed from the step equivalent to the system dimension to infinity. 
This formulation also admits an equivalent finite-horizon implementation.
Main \emph{contributions} of this paper include: 
\begin{itemize}
\item An infinite-horizon finite-time MPC framework is proposed, avoiding terminal equality constraints, short control horizon, and switching strategies.

\item 
%equivalent finite-horizon implementation that ensures both asymptotic stability and finite-time convergence, with significantly enlarged feasibility region;
The proposed infinite-horizon MPC can be implemented equivalently using a conventional finite-horizon MPC with a terminal cost and an invariant terminal set. This provides a practical implementation without changing the finite-time stabilization property.

\item 
The proposed finite-time MPC can be systematically extended for multi-input and nonlinear systems.
%validated by numerical examples.
%Three numerical examples are provided to illustrate the theoretical results.
\end{itemize} 
Numerical examples are presented to demonstrate the theoretical findings.

%The rest of this paper is arranged as follows:
%The problem of constrained finite-time control is formulated in Section \ref{sec problem}.
%The proposed finite-time MPC with infinite control horizon is designed in Section \ref{sec main}, where extensions to multi-input systems and nonlinear systems are described.
%Some simulation examples are provided in Section \ref{sec sim}.
%Concluding remarks are given in the final section.

%The paper is organized as follows. Section \ref{sec problem} formulates the constrained finite-time control problem. Section \ref{sec main} develops the finite-time MPC with an infinite horizon and its extensions to multi-input and nonlinear systems. Simulation results are reported in Section \ref{sec sim}, and concluding remarks are given in the last section.

%\vspace{5mm}
%-------------------------------------------------------------------
\section{Problem Formulation}\label{sec problem}

{
\subsection{Notation}

In this paper, $\mathbb{Z}^+$ denotes the set of all positive integers;
%$\mathbb{C}$ denotes the set of all complex numbers;
and $\mathbb{R}^n$ denotes the set of all $n$-dimensional real-value vectors, where $n\in \mathbb{Z}^+$.

The time instant is denoted by $k\in \mathbb{Z}^+\cup\{0\}$.
Throughout this paper, $k=0$ is the initial time.

The system state and control are denoted by $x(k) \in \mathbb{R}^n$ and $u(k) \in \mathbb{R}$, respectively. Predicted trajectories are collected as
\begin{align}\label{predictive x}
	X(k) %=&  [x^T(1|k), x^T(1|k), \cdots, x^T(N|k)]^T,\\
    =&[x^T(1|k), x^T(2|k), \cdots, x^T(N|k)]^T,\\
%    \end{align}
%and
%\begin{align}
\label{predictive u}
	U(k) %=&  [u(0|k), u(1|k), \cdots, u(N-1|k)]^T,\\
    =& [u(0|k), u(1|k), \cdots, u(N-1|k)]^T.
\end{align}
where $x(i|k)$ and $u(i|k)$ denote the $i$th-step predictive state and control from time $k$; and $N$ denotes the control horizon.
The optimal state and control are denoted by $x^\ast(i|k)$ and $u^\ast(i|k)$, respectively; and the feasible (sub-optimal) state and control are represented by $\tilde x(i|k)$ and $\tilde u(i|k)$, respectively.

State and control constraints are denoted by $\mathcal{X}$ and $\mathcal{U}$, respectively. The terminal constraint is represented by $\mathcal{X}_f$.

The notation $\otimes$ (or $\bigotimes$) represents Cartesian product.
%and $\oplus$ denotes the Minkowski addition.

}%end of blue

\subsection{Existing result on finite-time MPC}

Consider the discrete-time linear plant
\begin{align}\label{linear syst}
x(k+1) = A x(k) + b u(k),
\end{align}
subject to
\begin{align}\label{constraints}
u \in \mathcal{U},~~x \in \mathcal{X},
\end{align}
where $x \in \mathbb{R}^n$ and $u \in \mathbb{R}$ are the state and the control input. Here, $n$ represents the dimension of the overall system, and $k$ represents the discrete time instant. The sets $\mathcal{U}$ and $\mathcal{X}$ are convex constraints that include the origin.

\begin{ass}
    The pair $(A,b)$ is fully controllable, or equivalently, the matrix
    \begin{align*}
        S = [A^{n-1}b, A^{n-2}b, \cdots, Ab, b]
    \end{align*}
    is non-singular.
\end{ass}

In \cite{zhu2024finite}, it was proposed and proved that finite-time stabilization can be guaranteed by the following MPC with 1) penalizing only the terminal cost, and 2) setting the control horizon equal to the system dimension:
\begin{align}\label{deadbeat opt pre}
	[U^\ast(k), ~X^\ast(k)] = \mathrm{arg}\min_{U(k), ~X(k)}\|x(n|k)\|^2_P,
\end{align}
subject to 
\begin{align}
	&x(0|k)=x(k), \label{initial constraint}\\
    &x(i+1|k) = Ax(i|k) + bu(i|k),~i=0,\cdots,n-1,  \label{dyn constraint}\\
	&u(i|k)\in\mathcal{U},~x(i+1|k)\in\mathcal{X}, ~i=0,\cdots,n-1,  \label{state and control constraint n}\\
	&x(n|k)\in \mathcal{X}_f, \label{terminal constraint n}
\end{align}
%where $x(i|k)$ and $u(i|k)$ are the $i$-th predictive state and control at time $k$, and
%where $x(i|k)$ and $u(i|k)$ denote the $i$th-step predicted state and control from time $k$. The predicted trajectories are collected as
%\begin{align*}
%	X(k) %=&  [x^T(1|k), x^T(1|k), \cdots, x^T(N|k)]^T,\\
%    =&[x^T(1|k), x^T(2|k), \cdots, x^T(n|k)]^T,
%    \end{align*}
%and
%\begin{align*}
%	U(k) %=&  [u(0|k), u(1|k), \cdots, u(N-1|k)]^T,\\
%    =& [u(0|k), u(1|k), \cdots, u(n-1|k)]^T.
%\end{align*}
%are the predictive state sequence and predictive control sequence, respectively;
where the weighting matrix $P$ in \eqref{deadbeat opt pre} is solved from 
\begin{align}\label{Lyapunov eqn}
	(A-bK)^TP(A-bK)-P = -Q-K^TRK,
\end{align}
with given positive matrices $Q$ and $R$;
the feedback gain $K$ is chosen such that $|\mathrm{eig}(A-bK)|<1$.
%the control horizon is set equal to the system dimension, i.e., $N=n$.
The finite-time MPC is executed on the linear plant \eqref{linear syst} via receding horizon scheme $$u(k)= [1, 0, \cdots,0]_{1\times n}U^\ast(k),$$
where $U^\ast(k)$ is solved from constrained optimization \eqref{deadbeat opt pre}.

Moreover, if constraints \eqref{constraints} are not considered, the solution to the unconstrained finite-time MPC is
\begin{align*}
    u(k) = -K_{db}x(k),
\end{align*}
ensuring that all eigenvalues of $(A-bK_{db})$ are zero.

The above finite-time MPC does not depend on switching inside one-step region \cite{anderson2018finite} or terminal equality constraint \cite{scokaert1999suboptimal}. However, its control horizon is short, sometimes leading to limited initial feasibility region.

\begin{remark}
    The terminal weighting matrix $P$ in \eqref{Lyapunov eqn} is used to guarantee asymptotic stability. In \cite{zhu2024finite}, the finite-time stability is based on asymptotic stability, with the control horizon set equal to the system dimension.
\end{remark}

%\begin{remark}
%    In the unconstrained case, the (unique) solution to \eqref{deadbeat opt pre} is $x(n|k)=0$, if $N=n$ is assigned. Comparatively, if stage costs (other than the terminal cost in \eqref{deadbeat opt pre}) are also penalized, then $x(n|k)=0$ is not unique, and finite-time convergence cannot be ensured.
%\end{remark}

%----------------------------------------
\subsection{Problem statement}

%The \emph{objective} is to design stabilizing MPC with finite-time convergence performance for the constrained linear plant \eqref{linear syst} subject to state and control constraints \eqref{constraints}.
%Compared with the previous terminal cost strategy \cite{zhu2024finite}, the initial feasibility of the proposed finite-time MPC in this paper should be significantly improved.

The \emph{objective} is to design a stabilizing MPC framework that achieves finite-time convergence property for the constrained linear system \eqref{linear syst}–\eqref{constraints}. Compare with the terminal-cost formulation in \cite{zhu2024finite}, the proposed scheme offers a significant improvement in initial feasibility.

%With modifications, the proposed finite-time MPC should also be applicable to 
With appropriate modifications, the finite-time MPC with improved initial feasibility can be extended to:
\begin{itemize}
	\item constrained multi-input linear systems 
\begin{align}\label{multi input syst}
	x(k+1) = %Ax(k)+B u(k),
     Ax(k) + \sum_{j=1}^mb_ju_j(k) = Ax(k)+Bu(k),
\end{align}
subject to
\begin{align}\label{constraints multi}
x\in\mathcal{X}\subset\mathbb{R}^n,~u\in\mathcal{U}\subset\mathbb{R}^m,~~m>1,
\end{align}
where $u=[u_1,\cdots,u_m]^T$; $B=[b_1,\cdots,b_m]$. 

\item constrained nonlinear plant
\begin{align}
	&x(k+1)= f(x(k), u(k)), \label{nonlinear syst}\\
&u\in\mathcal{U}\subset\mathbb{R},~x\in\mathcal{X}\subset\mathbb{R}^n \label{constraints nonl}
\end{align}
%where the nonlinear function $f: \mathbb{R}^n\times\mathbb{R}^1\to \mathbb{R}^n$  
%is continuously differentiable with respect to $x$ and $u$, and it satisfies $f(0,0)=0$. 
%It is assumed that the nonlinear plant is feedback linearizable \cite{aranda1996linearization} in a compact set $\{x\in \mathcal{D}\subseteq \mathcal{X}\}$ containing the origin.
%
where $f:\mathbb{R}^n\times\mathbb{R} \to \mathbb{R}^n$ is continuously differentiable in both arguments and satisfies $f(0,0)=0$. 
\end{itemize}

\begin{ass}\label{assumption controllability}
    The pair $(A, B)$ is fully controllable.
\end{ass}

\begin{ass}\label{assumption nonlinear}
    The nonlinear plant \eqref{nonlinear syst} is feedback linearizable in a compact set $\mathcal{D}$, i.e., there exists a nonlinear diffeomorphism $z=M(x)$ in $\mathcal{D}$ and a corresponding feedback $v=\mu(x,u)$, such that \eqref{nonlinear syst} can be transformed to linear form:
    \begin{align*}
        z(k+1)= Fz(k)+Gv(k),
    \end{align*}
    where the pair $(F,G)$ is controllable.
\end{ass}

It follows from Assumption \ref{assumption nonlinear} that, for the feedback linearizable nonlinear plant \eqref{nonlinear syst},
    there exists an invariant set $\mathcal{X}_f\subset \mathcal{X}$ containing the origin, and a corresponding feedback control $u=\kappa(x)$, such that
    \begin{align}\label{nonlinear inv}
        \kappa(x)\in \mathcal{U}, ~f(x,\kappa(x)) \in \mathcal{X}_f,~~\forall~x\in\mathcal{X}_f.
    \end{align}
    Moreover, inside $\mathcal{X}_f$, there exists a Lyapunov function $V(x)$, such that
    \begin{align*}
        V(f(x,\kappa(x))-V(x)\leq -l(x,\kappa(x)),
    \end{align*}
    where $l(x,\kappa(x))$ is a positive definite function.

\vspace{4mm}
%===============================================================
\section{Main results}\label{sec main}

The proposed finite-time MPC is designed within an infinite control horizon framework to improve initial feasibility.
It can be implemented via finite-horizon MPC with a sufficiently large control horizon.
%With modifications, the proposed finite-time MPC is applicable to multi-input systems and nonlinear systems.
%With appropriate adaptations, the proposed finite-time MPC can be applied to systems with multiple inputs as well as nonlinear dynamics.

%-----------------------------------------------------
\subsection{Finite-time MPC with infinite control horizon}

%In the proposed finite-time MPC with infinite control horizon, 
The \emph{key setting} in the finite-time MPC with infinite control horizon is that, summation of stage costs starts from the number of system dimension to infinity, i.e., the lower summation index is ``$i=n$'':
\begin{align}\label{cost fun}
    J(k) = \sum_{i=n}^{\infty}\left(\|x(i|k)\|^2_Q+\|u(i|k)\|_R^2\right).
\end{align}
The optimization in finite-time MPC is established by
\begin{align}\label{deadbeat opt inf}
	[U^\ast(k), X^\ast(k)] = \mathrm{arg}\min_{U(k), X(k)} J(k),
\end{align}
subject to \eqref{initial constraint}--\eqref{state and control constraint n} for infinite $i=0,1,\dots, \infty$.
%\begin{align}
%&x(0|k)=x(k),\\
%	&x(i+1|k) = Ax(i|k) + bu(i|k),~i=0,1,\cdots,\\
%	&u(i|k)\in\mathcal{U},~x(i|k)\in\mathcal{X}, ~i=0,1,\cdots, 
%\end{align}
The finite-time MPC is implemented by 
\begin{equation}\label{receding horizon imp}
    u(k) = [1,0,\cdots,\cdots]_{1\times \infty}U^\ast(k), %=u^\ast(0|k),
\end{equation}
where the optimal predictive control sequence $U^\ast(k)$ is solved from \eqref{deadbeat opt inf}.

\begin{remark}
	%Compared with \cite{zhu2024finite}, 
 In \eqref{deadbeat opt inf}, the control horizon is infinity, such that the initial feasibility region is comparable with classic MPC techniques.
\end{remark}

\begin{lemma}\label{lemma1}
Suppose that $(A,b)$ is controllable. In the absence of constraints, control $u(k)=-K_{db}x(k)$ with $K_{db}=[1,0,\cdots,0]S^{-1}A^n$ ensures that the state of
 \begin{align*}
     x(k+1) = Ax(k)+bu(k) = (A-bK_{db})x(k)
 \end{align*}
 converges to the origin in finite $n$ time steps.
\end{lemma}

\textbf{\textit{Proof:}} 
Suppose that the control sequence in the next $n$ steps is
\begin{align*}
    U_{[0,n-1]}(k) \triangleq [u(0|k),u(1|k),\cdots,u(n-1|k)]^T,
\end{align*}
then
\begin{align*}
    x(n|k) = A^n x(0|k) + SU_{[0,n-1]}(k),
\end{align*}
where $x(0|k)=x(k)$ and $S=[A^{n-1}b,\cdots,Ab,b]$.

For $x(n|k)=0$, the optimal control sequence can be uniquely solved by
\begin{align*}
    U_{[0,n-1]}(k) = -S^{-1}A^n x(k),
\end{align*}
and
\begin{align*}
    u(k) = -[1,0,\cdots,0]S^{-1}A^nx(k)=-K_{db}x(k),
\end{align*}
such that $J(k)=0$ reaches its minimum.
The inverse $S^{-1}$ exists due to controllability of $(A,b)$.
This completes the proof. \hfill $\square$

%The finite-time convergence performance is given in the following theorem.

\begin{theorem}\label{Thm 1}
Consider the linear discrete-time system \eqref{linear syst} subject to constraints \eqref{constraints}.  Design the cost function by \eqref{cost fun}, where $n$ is the system dimension, and the sum of stage costs is from $n$ to $\infty$.
	If the constrained optimization \eqref{deadbeat opt inf} subject to constraints \eqref{initial constraint}--\eqref{state and control constraint n} is feasible initially,
    and the control is implemented by \eqref{receding horizon imp},
	then,
	\begin{enumerate}
		\item[1)] the constrained optimization \eqref{deadbeat opt inf} subject to \eqref{initial constraint}--\eqref{state and control constraint n} is  feasible recursively;
		%\item[2)] the closed-loop system 
        %with the proposed finite-time MPC \eqref{deadbeat opt inf}--\eqref{receding horizon imp} 
        %is asymptotically stable;
		\item[2)] there exist a finite $T>0$, such that the closed-loop state $x(k)=0$ whenever $k>T$.
	\end{enumerate}
\end{theorem}

\textbf{\textit{Proof:}}
	%The statement 1) is obvious by using ``tail" method.
        1) Suppose, at time $k$, the constrained optimization \eqref{deadbeat opt inf} is feasible, and the optimal predicted sequences are
        \begin{align*}
            x^\ast(i|k),~u^\ast(i|k),~i=0,1,\cdots,\infty,
        \end{align*}
        satisfying system dynamics and all state and control constraints.
        
        Then, at time $k+1$, a feasible candidate solution can be built by shifting the original sequences by one step:
        \begin{align*}
            \tilde{x}(i|k+1) = x^\ast (i+1|k),~
        \tilde{u}(i|k+1)= u^\ast (i+1|k),
        \end{align*}
        for $i=0,1,\cdots,\infty$, satisfying system dynamics and all constraints.
	
	2) %For $i=1,2,\cdots$,  let $x^\ast(i|k)$ and $u^\ast(i-1|k)$ denote the optimal predictive state sequence and control sequence at time $k$.
	%\begin{align*}
	%	J(k) \triangleq \sum_{i=n}^{\infty}\|x(i|k)\|^2_Q+\|u(i|k)\|_R^2.
	%\end{align*}
	%Then feasible state sequence and control sequence at $k+1$ exist:
	%\begin{align*}
	%	\tilde{x}(i|k+1) = x^\ast (i+1|k),\\
    %    \tilde{u}(i-1|k+1)= u^\ast (i|k).
	%\end{align*}
	Denote the optimal cost at time $k$ by
	\begin{align*}
		J^\ast(k) =& \sum_{i=n}^{\infty}\|x^\ast(i|k)\|^2_Q+\|u^\ast(i|k)\|_R^2,
        \end{align*}
        and
        \begin{align*}
		\tilde{J}(k+1) =& \sum_{i=n}^{\infty}\|\tilde{x}(i|k+1)\|^2_Q+\|\tilde{u}(i|k+1)\|_R^2\\
        =&\sum_{i=n}^{\infty}\|{x}^\ast(i+1|k)\|^2_Q+\|{u}^\ast(i+1|k)\|_R^2.
	\end{align*}
	It holds that
 \begin{align*}
		J^\ast(k+1)-J^\ast(k) \leq&  \tilde{J}(k+1)-J^\ast(k) \\
		\leq & -{x^\ast}^T(n|k)Qx^\ast(n|k)\leq 0,
\end{align*}
	indicating that $x^\ast(n|k)$ converges to the origin as $k\to \infty$. 
    
    In another aspect, consider
	\begin{align}
		x^\ast(n|k) =& Ax^\ast(n-1|k)+bu^\ast(n-1|k)\\
		=&A^n x(k) + \sum_{i=0}^{n-1} A^i bu^\ast(n-1-i|k)\\
		=&A^n x(k) + %[A^{n-1}b,~A^{n-2}b,~\cdots,~b] 
        SU^\ast_{[0:~n-1]}(k), \label{last row of xn}
	\end{align}
	where $S\triangleq [A^{n-1}b,\cdots,b]$ is non-singular due to controllability of $(A,b)$, 
    and
	\begin{align*}
		U^\ast_{[0:~n-1]}(k) \triangleq  [u^\ast(0|k),\cdots,u^\ast(n-1|k)]^T.
	\end{align*}
	From \eqref{last row of xn}, the optimal control sequence can be solved by
	\begin{align*}
		U^\ast_{[0:~n-1]}(k)= S^{-1} \left( x^\ast(n|k)-A^nx(k)\right),
	\end{align*}
	such that
	\begin{align}
		u(k)=& u^\ast(0|k) =[1,0,\cdots,0]_{1\times n}U^\ast_{[0:~n-1]}(k)\\
		=&[1,0,\cdots,0]_{1\times n}S^{-1} \left( x^\ast(n|k)-A^nx(k)\right). \label{converging u}
	\end{align} 
	Substituting \eqref{converging u} into \eqref{linear syst} yields
	\begin{align}\label{closed-loop}
		x(k+1) = (A-bK_{db})x(k)+ BS_n^T x^\ast(n|k),
	\end{align}
	where $K_{db}=[1,0,\cdots,0]S^{-1}A^n$ is the deadbeat gain (Lemma \ref{lemma1}), i.e., all eigenvalues of $(A-bK_{db})$ are zero; and $S_n^T=[1,0,\cdots,0]S^{-1}$.
	
    It can be seen that \eqref{closed-loop} is an asymotptic stable linear system perturbed by an asymptotically vanishing term; consequently, $x(k)=0$ is asymptotically stable.

 %To prove 3), 
 %consider that the closed-loop state converges to zero asymptotically, such that it converges, in some finite time $T_1>0$, to a sufficiently small set (containing the origin) where constraints \eqref{constraints} are naturally satisfied, and the constrained optimization actually becomes an unconstrained optimization. Then, the optimal solution is
 3) Note that $x(k)$ approaches the origin asymptotically. Hence, after some finite time $T_1>0$, the state enters sufficiently small set where constraints \eqref{initial constraint}--\eqref{state and control constraint n} are fulfilled by
	\begin{align*}
		u^\ast(i|k)= -K_{db}x^\ast(i|k),%~i=0,1,\cdots,
	\end{align*}
such that the constrained optimization is reduced to unconstrained optimization. 
It then follows from Lemma \ref{lemma1} that
%It then holds that
    \begin{align*}
        x^\ast(i|k)=0,~\forall~i\geq n, %n+1, n+2\cdots.
    \end{align*}
%and $J(k)=0$ in \eqref{cost fun} reaches its minimum.
Consequently,
    \begin{align*}
        u(k) = u^\ast (0|k) = -K_{db}x^\ast(0|k) = -K_{db}x(k),
    \end{align*}
	such that %$x^\ast(i|k)=0$ for all $i\geq n$.
	%the closed-loop system satisfies
    \begin{align*}
        x(k+1) = (A-bK_{db})x(k),
    \end{align*}
    and its state converges to the origin in finite time $T=T_1+n< +\infty$.
	\hfill $\square$

%\begin{remark}
%    It should be noted that, similar to classical MPC formulations, the results of Theorem 1 are established under the assumption of initial feasibility, i.e., the optimization problem admits a feasible solution at the initial step. This assumption is standard in MPC literature, since recursive feasibility and stability can only be guaranteed once the optimization is feasible initially. Importantly, the proposed infinite-horizon finite-time MPC enlarges the initial feasible set significantly compared with previous terminal-cost strategies, as demonstrated in Fig. 2 in simulation section. Hence, while initial feasibility is assumed for theoretical analysis, the proposed framework improves the practical likelihood of satisfying this assumption.
%\end{remark}

%--------------------------------------------------------------------
\subsection{Implementation via finite-horizon MPC}

%In this section, an implementation via finite-horizon MPC is proposed to guarantee the finite-time performance of the proposed closed-loop system.

Set a finite control horizon $N> n$, and design the finite-horizon cost function by
\begin{align}\label{cost finite-horizon}
    J_f(k) = \sum_{i=n}^{N-1}\left(\|x(i|k)\|^2_Q+\|u(i|k)\|_R^2\right)
	+\|x(N|k)\|^2_P,
\end{align}
where the summation is from $n$ to $N-1$;
the weighting matrix $P$ is solved from Lyapunov equation \eqref{Lyapunov eqn}.
%with $K$ satisfying that all eigenvalues of $(A-bK)$ are inside the unit circle on the complex plane.

The finite-time MPC is constructed in the implementable finite-horizon framework:
\begin{align}
	\begin{split}\label{deadbeat opt finite}
		[U^\ast(k), ~X^\ast(k)] =& \mathrm{arg}\min_{U(k),~X(k)} J_f(k),
	\end{split}
\end{align}
subject to constraints
\eqref{initial constraint}--\eqref{state and control constraint n} for finite $i=0,\dots, N-1$, and %the terminal constraint
\begin{align}
%&x(0|k)=x(k),\\
%	&x(i+1|k) = Ax(i|k) + bu(i|k),~i=0,\cdots N-1,\\
%	&u(i-1|k)\in\mathcal{U},~x(i|k)\in\mathcal{X}, ~i=0,\cdots,N,\\
	x(N|k) \in \mathcal{X}_f, \label{terminal constraint N}
\end{align}
where %$N$ is the finite control horizon satisfying $N>n$;
$\mathcal{X}_f$ is an invariant terminal set satisfying
\begin{align}
	\mathcal{X}_f\subset \mathcal{X}, %\label{inv constraint}\\
	~~~-K\mathcal{X}_f\subset \mathcal{U},  % \label{inv control constraint}\\
	~~~(A-bK)\mathcal{X}_f \subset \mathcal{X}_f. \label{inv state constraint}
\end{align}
Then, the finite-time MPC is implemented via
\begin{align}\label{receding horizon imp 2}
    u(k) = [1,0,\cdots,0]_{1\times N}U^\ast(k), %=u^\ast (0|k),
\end{align} 
where $U^\ast(k)$ is solved from the optimization \eqref{deadbeat opt finite}.

\begin{remark}
Feedback gain $K$ in \eqref{inv state constraint} does not have to be the deadbeat gain $K_{db}$ or LQR gain $K_{LQR}$.
\end{remark}
\begin{remark}
The terminal constraint can be obtained by
\begin{align}\label{calculation of terminal set}
  \mathcal{X}_f=\{x\in \mathbb{R}^2 \left| x^TPx\leq \epsilon \right.\},
\end{align}
where $\epsilon>0$ is solved from quadratic programming:
\begin{align}\label{quadratic prog of terminal set}
   \epsilon = \max_{x\in \mathcal{X},~Kx\in \mathcal{U}} x^TPx.
\end{align}
\end{remark}

\begin{theorem}\label{thm finite}
	Consider the linear discrete-time plant \eqref{linear syst} subject to constraints \eqref{constraints}, and suppose that $(A,b)$ is controllable.
    Design the cost function by \eqref{cost finite-horizon}, where the summation is from $n$ to $N-1$,
    and the terminal weight $P$ is solved from Lyapunov equation \eqref{Lyapunov eqn}.  
    If the constrained optimization \eqref{deadbeat opt finite} 
    (subject to \eqref{initial constraint}--\eqref{state and control constraint n} and \eqref{terminal constraint N}) 
    is feasible initially, and the MPC is implemented by \eqref{deadbeat opt finite}--\eqref{receding horizon imp 2},
	then,
	\begin{enumerate}
	    \item[1)] the finite-horizon cost function \eqref{cost finite-horizon} is equivalent to the infinite-horizon cost function \eqref{cost fun};
        \item[2)] the constrained optimization \eqref{deadbeat opt finite} (subject to \eqref{initial constraint}--\eqref{state and control constraint n} and \eqref{terminal constraint N}) 
        is feasible recursively;
		%\item[2)] the closed-loop system %with the finite-time MPC \eqref{deadbeat opt finite}--\eqref{receding horizon imp 2} 
        %is asymptotically stable;
		\item[3)] there exist a finite $T>0$, such that the closed-loop state $x(k)=0$ whenever $k>T$.
	\end{enumerate}
\end{theorem}

\textbf{\textit{Proof:}}
For $N>n$, the infinite-horizon cost function \eqref{cost fun} can be decomposed by
\begin{align*}
    J(k) = \sum_{i=n}^{N-1}\|x(i|k)\|_Q^2+\|u(i|k)\|^2_R
    +\sum_{i=N}^{+\infty}\|x(i|k)\|_Q^2+\|u(i|k)\|^2_R.
\end{align*}

Since the optimization \eqref{deadbeat opt finite} 
    (subject to \eqref{initial constraint}--\eqref{state and control constraint n} and \eqref{terminal constraint N}) is feasible initially, it follows that the system state enters the terminal set \eqref{terminal constraint N} after $i=N$, such that state and control constraints are satisfied naturally by the linear feedback control $u(i|k)=-Kx(i|k)$, where $K$ is solved from the discrete-time Lyapunov equation \eqref{Lyapunov eqn}. Existence of the solution $K$ to \eqref{Lyapunov eqn} is guaranteed by controllability of $(A,b)$.

    It then follows from \eqref{Lyapunov eqn} that, inside the terminal set (from $i=N$ to $\infty$), we have
    \begin{align*}
        x^T(i|k)(A-bK)^T&P(A-bK)x(i|k)-x^T(i|k)Px(i|k)\\
        =&-x^T(i|k)Qx(i|k)-x^T(i|k)K^TRKx(i|k),
    \end{align*}
    such that
\begin{align*}
    x^T(i+1|k)P&x(i+1|k) -x^T(i|k)Px(i|k)\\
    &= -x^T(i|k)Qx^T(i|k)-u^T(i|k)Ru(i|k).
\end{align*}
Summation of the above equation from $i=N$ to $\infty$ yields
\begin{align*}
    \|x(N|k)\|^2_P = \sum_{i=N}^\infty \left(\|x(i|k)\|^2_Q + \|u(i|k)\|^2_R\right),
\end{align*}
and therefore
\begin{align*}
	{J}^\ast_f(k)&\leq J_f(k) = J(k)\\
    &= \sum_{i=n}^\infty \left(\left\|x(i|k)\right\|^2_Q+\left\|u(i|k)\right\|_R^2\right).
\end{align*}
The rest of the proof can be completed by following the steps in Theorem \ref{Thm 1}.
\hfill $\square$

\begin{remark}
    %In the proof of Theorem \ref{Thm 1}, only the fact of finite-time convergence is proved. However, 
    %Explicit calculation of the finite convergence time depends on detailed constraints, and it is not addressed in Theorem \ref{Thm 1}.
    If the optimization is feasible initially, the convergence time can be conservatively estimated as following: 
    outside the terminal set, a feasible control sequence exists such that state enters the terminal set before $N$ time steps. Inside the terminal set where the state and input constraints are actually inactive, the state converges to the origin in $n$ time steps. Consequently, the closed-loop state x(k) converges to the origin before $N+n$ steps.
\end{remark}

\begin{remark}
Although it is finite, the control horizon $N$ can be chosen sufficiently large to improve initial feasibility.
\end{remark}

\begin{remark}
    By replacing the equality terminal constraint in \cite{zhu2024finite} with the inequality terminal constraint in this paper, such that the control horizon can be set independent of the system dimension.
    Results in \cite{zhu2024finite} can be regarded as the special case of Theorem \ref{thm finite} with $N=n+1$.
\end{remark}

%\begin{remark}
%Neither terminal equality constraints nor switching techniques are applied in the implementable finite-time MPC.
%\end{remark}

\begin{remark} 
    Compared with controller matching \cite{di2009model, dreke2023practical}, the structure of our proposed finite-time MPC framework is fundamentally different: our cost function contains stage costs starting from predictive step $n$. This leads to different proof of closed-loop stability and finite-time convergence property. Particularly, in \cite{di2009model, dreke2023practical}, the final result is practical finite-time convergence (convergence into a neighborhood of the origin in finite-time).
\end{remark}

The running steps of the proposed finite-time MPC can be summarized in Algorithm 1.

\begin{algorithm}
	%\textsl{}\setstretch{1.8}
	\renewcommand{\algorithmicrequire}{\textbf{Offline:}}
	\renewcommand{\algorithmicensure}{\textbf{Online:}}
	\caption{\color{blue} Constrained finite-time MPC for single-input discrete-time linear systems}
	\label{alg1}
	\begin{algorithmic}[1]
	\REQUIRE
	\STATE	Set the control horizon greater than the system dimension, i.e. $N>n$; set $Q$ and $R$ to be positive definite matrices.
	\STATE  Find a feedback gain $K$ such that eigenvalues of $A-BK$ are inside the unite circle.
	\STATE Solve the terminal weighting matrix $P$ from Lyapunov equation \eqref{Lyapunov eqn}.
	\STATE Calculate the terminal constraint via \eqref{calculation of terminal set} and \eqref{quadratic prog of terminal set}.
	\ENSURE
	\FOR{$k=0$ to the end of task}
	\STATE Measure the system state $x(k)$.
	\STATE Solve the optimization \eqref{deadbeat opt finite} subject to constraints \eqref{initial constraint}--\eqref{state and control constraint n} and \eqref{terminal constraint N} to obtain the optimal control sequence $U^\ast(k)$.
	\STATE Apply $u(k)=[1,0,\cdots,0]U^\ast(k)$ to system \eqref{linear syst}.
	\ENDFOR
	\end{algorithmic} 
\end{algorithm}

%-----------------------------------------------------------------------
\subsection{Extension to multi-input linear systems}

The controllable multi-input system \eqref{multi input syst}
%\begin{align}
	%\begin{split}
%		\label{multi input syst}
%		x(k+1) =
%		Ax(k) + \sum_{j=1}^m b_ju_j(k),
	%\end{split}
%\end{align}
%subject to constraints \eqref{constraints multi}.
%$m>1$ denotes the number of inputs, and 
%$b_j$ denotes the $j$-th column of the matrix $B$. 
%The pair $(A, B)$ is assumed to be fully controllable.
%
%The controllable multi-input system \eqref{multi input syst} 
can be transformed into a decoupled form \cite{Wonham1974Linear}:
\begin{align}\label{decouled syst}
	z(k+1) = Fz(k) +G u(k),
\end{align}
where $z = Mx$ is the invertible linear transformation, and
\begin{align*}
	F =& MAM^{-1} = \begin{bmatrix}
		F_{11} &F_{12} &\cdots &F_{1q}\\
		0 &F_{22} &\cdots &F_{2q}\\
		\vdots & &\ddots &\vdots\\
		0 &0 &\cdots &F_{qq}
	\end{bmatrix},\\
G =& MB = \left[\begin{array}{c:c:c:c}
	g_1 & & & \ast\\
	\hdashline
	0 &\ddots & & \ast\\
	\hdashline
	0 &0 &g_q & \ast
\end{array}\right],
\end{align*}
where, for $j=1,\cdots,q$, $F_{jj}\in\mathbb{R}^{n_j\times n_j}$, $g_j\in\mathbb{R}^{n_j\times 1}$, $\sum_{i=j}^{q}n_j=n$, and $(F_{jj}, g_j)$  are controllable. 
Denote $z = [z_1^T, \cdots, z_q^T]^T$, where $z_j \in \mathbb{R}^{n_j}$ for $j=1,\cdots,q$.
Denote $u = [u_1,\cdots,u_q,u_{q+1},\cdots,u_m]^T$.
The controllable multi-input system can be regarded as $q$ subsystems,
with $u_j ~(j=1,\cdots, q)$ being the single input for the $j$-th subsystem,
and $u_l~(l=q+1,\cdots, m)$ are set to zero. 

The cost function can be designed by
\begin{align}\label{opt multi}
	J_m = \sum_{j=1}^{q}\left( \sum_{i=n_j}^{N-1} \left(\|z_j(i|k)\|^2_{Q_j} + \|u_j(i|k)\|^2_{R_j}\right) + \|z_j(N|k)\|^2_{P_j}\right)
\end{align}
where $N$ is the control horizon satisfying $N>n_j$ for all $j=1,2,\cdots, q$;
$Q_j$ and $R_j$ are positive definite matrices; $P_j$ is solved from Lyapunov equation
\begin{align*}
	(F_{jj}-g_jK_j)^TP_j(F_{jj}-g_jK_j)-P_j = -Q_j-K_j^TR_jK_j,
\end{align*}
with $K_j$ selected such that all eigenvalues of $(F_{jj}-g_jK_j)$ are inside the unit circle for all $j=1,2,\cdots,q$.

%\begin{align*}
%	J_i = \sum_{j=n_i}^{N-1} \|z(j|k)\|^2_{Q_i} + \|u(j|k)\|^2_{R_i} + \|z(N|k)\|^2_{P_i}
%\end{align*}

Define predictive control sequence by
\begin{align*}
	U_j(k) = [u_j(0|k),\cdots, u_j(N-1|k)]^T
\end{align*}
for all $j=1,2,\cdots,q$, and set $u_l=0$ for $l=q+1, \cdots, m$.
The optimization for the multi-input finite-time MPC can be constructed by
\begin{align}\label{multi opt}
	[U_j^\ast(k), Z_j^\ast(k)]_{j=1,\cdots,q}=\mathrm{arg}\min_{[U_j(k), Z_j(k)]_{j=1,\cdots,q}} J_m(k)
\end{align}
subject to
\begin{align}
	&z(i+1|k) = Fz(i|k)+Gu(i|k),~~\mbox{for}~i=0,\cdots,N-1,\\
	&z(i+1|k)\in M\mathcal{X},~~u(i|k)\in\mathcal{U},~~\mbox{for}~i=0,\cdots,N-1,\\
	& z(0|k) = Mx(k), \\
	&z_j(N|k)\in \mathcal{Z}_{f,j},~~\mbox{for}~j=1,\cdots,q,  \label{terminal constraint multi}
\end{align}
where $\mathcal{Z}_{f,j}$ are terminal sets satisfying
\begin{align}
	&\bigotimes_{j=1}^q\mathcal{Z}_{f,j}\subset M\mathcal{X}, %\\
	~~\bigotimes_{j=1}^q-K_{j}\mathcal{Z}_{f,j}\subset \mathcal{U}, \label{individual terminal sets 1} \\
	&(F_{jj}-g_jK_{j})\mathcal{Z}_{f,j} \in \mathcal{Z}_{f,j}.   \label{individual terminal sets 3}
\end{align}

The finite-time MPC for multi-input discrete-time linear systems can be implemented by
\begin{align}\label{RH control multi}
	u_i(k) &= [1,0,\cdots,0]_{1\times N}U_i^\ast(k), ~~\mbox{for}~i=1,\cdots,q,\\
	u_l(k) & =0,~~\mbox{for}~l=q+1,\cdots,m.   \label{redundence control}
\end{align}

\begin{theorem}\label{thm MI}
	For multi-input system \eqref{multi input syst} subject to constraints \eqref{constraints multi}, the MPC is designed by 
	\eqref{multi opt}--\eqref{redundence control}. 
 If the constrained optimization \eqref{multi opt}--\eqref{terminal constraint multi} is feasible initially, then the optimization is recursively feasible, and a finite time $T>0$ exists such that the closed-loop $x(k)=0$ for all $k>T$ .
\end{theorem}

\textbf{\textit{Proof:}}
%The proof can be completed by combining the proofs of Theorem \ref{thm finite} in this paper and Theorem 3 in \cite{zhu2024finite}. It is omitted here.
%
Consider the dynamics of $z_q$ in \eqref{decouled syst}:
\begin{align*}
    z_q(k+1) = F_{qq}z_q(k) + g_q u_q(k).
\end{align*}
Provided that the constrained optimization \eqref{multi opt}--\eqref{terminal constraint multi} is feasible initially,
it follows from Theorem \ref{thm finite} that $z_q$ converges to zero with finite steps. 

Then, the dynamics of $z_{q-1}$ becomes
\begin{align*}
    z_{q-1}(k+1) = F_{q-1,q-1}z_{q-1}(k) + g_{q-1} u_{q-1}(k),
\end{align*}
and Theorem \ref{thm finite} applies such that $z_{q-1}$ converges to zero with finite steps. 

The above process repeats, and all $z_j~(j=1,2,\cdots,q)$ converge to zero with finite steps.
\hfill $\square$

%-------------------------------------------------------------------------

\subsection{Extension to constrained nonlinear systems}

Based on Assumption \ref{assumption nonlinear}, the finite-horizon cost function for the nonlinear system \eqref{nonlinear syst} can be designed by
\begin{align*}
    J_{nl}(k) = \sum_{i=n}^{N-1}l(x(i|k), u(i|k)) + V(x(N|k)),
\end{align*}
where the positive definite functions $l(x,u)$ and $V(x)$ satisfy Assumption \ref{assumption nonlinear}.

The optimization for the nonlinear finite-time MPC can be formulated by
\begin{align}
	\begin{split}\label{deadbeat opt nonlinear}
		[U^\ast(k), ~X^\ast(k)] =
  \mathrm{arg}\min_{U(k), ~X(k)}& J_{nl}(k),
	\end{split}
\end{align}
subject to the initial constraint \eqref{initial constraint} and
\begin{align}
%& x(0|k) = x(k), \\
	&x(i+1|k) = f(x(i|k), u(i|k)), ~~\mbox{for}~i=0,1,\cdots, N-1,      \\
    & u(i|k)\in \mathcal{U},~~\mbox{for}~i=0,1,\cdots,N-1, \label{nonlinear control constraint}\\
	& x(i|k)\in \mathcal{X}\cap\mathcal{D},~~\mbox{for}~i=1,\cdots, N, \label{nonlinear state constraint}\\
	& x(N|k) \in \mathcal{X}_f, \label{terminal nonlinear}
\end{align}
where
$\mathcal{X}_f\subset \mathcal{X}$ is the terminal invariant set,  %$\mathcal{X}_f\subset\mathcal{X}\cap \mathcal{D}$, and 
%for all $x\in\mathcal{X}_f$, 
where a control $u=\kappa(x)$ exists and satisfies \eqref{nonlinear inv}.  %all eigenvalues of $(A-bK)$ are inside the unit circle, and

%The matrix $P$ is solved from Lyapunov equation \eqref{Lyapunov eqn} by using Jacobian matrices:
%\begin{align*}
%    A = \left.\frac{\partial f}{\partial x}\right|_{x=0,u=0},~b = \left.\frac{\partial f}{\partial u}\right|_{x=0,u=0}.
%\end{align*}

The nonlinear finite-time MPC is implemented by 
\begin{equation}\label{receding horizon imp nonl}
    u(k) = [1,0,\cdots,0]_{1\times N}U^\ast(k),
\end{equation}
where $U^\ast(k)$ is solved from nonlinear optimization \eqref{deadbeat opt nonlinear}. 

\begin{theorem}\label{thm nonlinear}
	Consider the nonlinear plant \eqref{nonlinear syst} subject to constraints \eqref{constraints}.
	%The constrained optimization is formulated by \eqref{deadbeat opt nonlinear}--\eqref{terminal nonlinear}.
	If the optimization \eqref{deadbeat opt nonlinear}--\eqref{terminal nonlinear} is feasible initially, and the implementation \eqref{receding horizon imp nonl} is applied, then
	\begin{enumerate}
		\item[1)] the optimization \eqref{deadbeat opt nonlinear}--\eqref{terminal nonlinear} is feasible recursively;
		%\item[2)] $x=0$ of the closed-loop system %with \eqref{deadbeat opt nonlinear}--\eqref{receding horizon imp nonl} 
        %is asymptotically stable;
		\item[2)] a finite time $T>0$ exists such that $x(k)=0$ for all $k>T$.
	\end{enumerate}
\end{theorem}

\textbf{\textit{Proof:}}
%The proof can be completed by combining the proofs of Theorem \ref{thm finite} in this paper and Theorem 4 in \cite{zhu2024finite}. It is omitted here.
%
%Consider the optimal cost function as the Lyapunov candidate:
%\begin{align*}
%    J^\ast(k) = \sum_{i=n}^{N-1}\left(\|x^\ast(i|k)\|^2_Q+\|u^\ast(i|k)\|_R^2\right)
%		 +\|x^\ast(N|k)\|^2_P.
%\end{align*}
%Since the nonlinear function $f(x,u)$ is differentiable with respect to $x$ and $u$, and $f(0,0)=0$, it holds that there exists constants $L_x$ and $L_u$ such that
%\begin{align*}
%\|x^\ast(1|k)\|^2=&\|f(x(k), u^\ast(0|k))\|^2\leq L_x\|x(k)\|^2+L_u\|u(0|k)\|^2,\\
%\|x^\ast(2|k)\|^2 =& \cdots \leq L_x\|x^\ast(1|k)\|^2 +L_u \|u(1|k)\|^2\\
%\leq & L_x^2\|x(k)\|^2 + L_xL_u\|u(0|k)\|^2 + L_u \|u(1|k)\|^2,\\
%&\vdots\\
%\|x^\ast(N|k)\|^2 \leq & L_x^N\|x(k)\|^2 + \sum_{i=0}^{N-1}L_x^{N-i-%1}L_u\|u(i|k)\|^2
%\end{align*}
%in the finite region $x\in \mathcal{X}$.
%
%Use the optimal state and control sequences as feasible state and control sequences at the next time instant:
%\begin{align*}
%    \tilde{x}(i|k+1) = x^\ast(i+1|k),~\tilde{u}(i|k+1) = u^\ast(i+1|k).
%\end{align*}
%Its variation can be calculated by
%\begin{align*}
%    J^\ast(k+1)-J^\ast(k)\leq \tilde{J}(k+1)-J^\ast(k)
%    = -\|x^\ast(n|k)\|^2_Q
%\end{align*}
%
1) Initial feasibility of optimization \eqref{deadbeat opt nonlinear}--\eqref{terminal nonlinear} at time $k$ indicates the existence of feasible control sequence and state sequence at time $k+1$:
\begin{align}\label{feasible solution nonlinear}
    {u}(i|k+1) = u^\ast(i+1|k),~{x}(i|k+1) = x^\ast(i+1|k),
\end{align}
for $i=0,\cdots,N-1$, and 
$$x(N|k+1)=f(x(N-1|k+1),\kappa(x(N-1|k+1))) \in \mathcal{X}_f$$ 
satisfying all constraints.

2) %It is claimed that $x^\ast(n|k)$ converges to the origin as $k\to +\infty$ by using the optimal cost function
%\begin{align*}
%    J_n^\ast(k) = \sum_{i=n}^{N-1}\left(\|x^\ast(i|k)\|^2_Q+\|u^\ast(i|k)\|_R^2\right)
%		 +\|x^\ast(N|k)\|^2_P
%\end{align*}
%as Lyapunov candidate.
Using the optimal cost function $J_{nl}^\ast(k)$ as the Lyapunov function yields
\begin{align}\label{decreasing Lya nonlinear}
    J_{nl}^\ast(k+1)-J_{nl}^\ast(k)\leq J_{nl}(k+1)-J_{nl}^\ast(k).
\end{align}
Substituting \eqref{feasible solution nonlinear} and the terminal constraint \eqref{terminal nonlinear} into \eqref{decreasing Lya nonlinear} yields
\begin{align*}
    J_{nl}^\ast(k+1)-J_{nl}^\ast(k) \leq -\|x^\ast(n|k)\|^2_Q-\|u^\ast(n|k)\|_R^2,
\end{align*}
implying that $x^\ast(n|k)$ converges asymptotically to the origin.
Convergence of $x^\ast(n|k)$ indicates that, for any $\epsilon>0$, there exists $T_\epsilon>0$, such that $\|x^\ast(n|k)\|<\epsilon$ for all $k>T_\epsilon$.

Asymptotic stability of $x=0$ can be proved by contradiction.
Suppose that $x(k)$ does not converge to zero as $k\to +\infty$, i.e,
for any finite time $T_\epsilon$, a constant $\epsilon_m>0$ exists such that 
\begin{align}\label{contradiction}
    \|x(k)\|\geq \epsilon_m,
\end{align}
with $k>T_\epsilon$.

Take $\epsilon<\epsilon_m$ and $k>T_\epsilon$. 
Based on initial and recursive feasibility, whenever $\|x(k)\|\geq \epsilon_m>0$, a feasible control sequence $U^\ast(k)$ exists such that 
\begin{align*} 
\|x^\ast(n|k)\|<\epsilon<\epsilon_m\leq \|x(k)\|,
\end{align*}
implying that feasible control
\begin{align*}
    u(k+i) = u^\ast(i|k),~i=0,\cdots,n-1,
  %  &u(k+1) = u^\ast(1|k),\\
  %  &\cdots\\
  %  & u(k+n-1)=u^\ast(n-1|k)
\end{align*}
drive the closed-loop system to $x(k+n)=x^\ast(n|k)$.
It indicates that, $\|x(k)\|< \epsilon$ for all $k>T_\epsilon+n$,
violating \eqref{contradiction}. 

As a consequence, it holds that $x=0$ of the closed-loop system is asymptotically stable.

Asymptotic stability of $x=0$ indicates that, before a finite time $T_2>0$, the state $x(k)$ enters the terminal invariant set, where constraints \eqref{nonlinear state constraint}--\eqref{terminal nonlinear} are inactive. The optimization \eqref{deadbeat opt nonlinear} is then equivalent to
\begin{align*}
		&[U^\ast(k), X^\ast(k)] =
  \mathrm{arg}\min_{U(k), Z(k)}\sum_{i=n}^{+\infty}l(x(i|k),u(i|k)),
\end{align*}
subject to
\begin{align}
    %&z(i+1|k) = Gz(i|k)+ Fv(i|k),
    &x(i+1|k)=f(x(i|k),u(i|k)),
    ~~\mbox{for}~i=0,1,\cdots, N-1,  \label{nonlinear dyn constraint}\\
	%& z(i|k) = M(x(i|k)),
    &x(0|k)=x(k). \label{nonlinear initial constraint}
\end{align}
The above constraints \eqref{nonlinear dyn constraint} and \eqref{nonlinear initial constraint} are equivalent to
\begin{align*}
    &z(i+1|k)=Fz(i|k)+Gv(i|k), ~~\mbox{for}~i=0,1,\cdots, N-1,\\
    &z(0|k) = M(x(k)), ~~v(i|k)=\mu(x(i|k),u(i|k)).
\end{align*}
It then follows from Lemma \ref{lemma1} that the feedback control
\begin{align*}
    v(i|k) = -[1,0,\cdots,0][F^{n-1}G,\cdots,FG,G]^{-1}F^nz(i|k)
\end{align*}
ensures the convergence of $z(k)$ to the origin in finite $n$ steps. Equivalently, the feedback control 
\begin{align*}
    u(i|k) = \mu^{-1}\left(z(i|k), v(i|k)\right)
\end{align*}
ensures the convergence of $x(k)$ to the origin in finite $n$ steps.

Consequently, the system state $x(k)$ converges to the origin in finite time $T=T_2+n$.
\hfill $\square$

\begin{remark}
	It is seen from \eqref{deadbeat opt nonlinear}--\eqref{terminal nonlinear} that, although the nonlinear diffeomorphism $z=M(x)$ and control $u=\mu(x,v)$ are used to prove the finite-time convergence property, they do not have to be calculated explicitly to design the finite-time MPC.
	%where a conservative small set can be found to be the terminal set $\mathcal{X}_f$.
\end{remark}

%\begin{remark}
%   Currently, extensions of results in \cite{di2009model, dreke2023practical} to nonlinear systems are not witnessed.
%\end{remark}

%\vspace{5mm}
%===============================================
\section{Numerical examples}\label{sec sim}

Numerical simulations are presented to verify the proposed finite-time MPC on single-input linear, multi-input linear, and nonlinear systems.

%----------------------------------------------
\subsection{Finite-time MPC for single input linear systems}\label{subsec sim single}

Consider the 2nd-order ($n=2$) plant:
\begin{align}\label{simulated A and b}
	A = \begin{bmatrix}
		1.1 &2\\ 0 &0.95
	\end{bmatrix},~~
	b = \begin{bmatrix}
		0\\ 0.079
	\end{bmatrix},
\end{align}
subject to control constraints $|u|\leq 5$.

The finite-time MPC is designed by following Theorem \ref{thm finite},
where the control horizon is selected by $N=8$; the weighting matrix are chosen by $Q = I_{2\times 2}$ and $R=0.1$; 
%the feedback gain is chosen by $K = [4.3~24.7]$ such that all eigenvalues of $(A-bK)$ are inside the unit circle;
a feedback gain $K = [4.3~24.7]$ is selected to ensure that $|\mathrm{eig}(A-BK)|<1$.
%the eigenvalues of $(A-bK)$ lie strictly inside the unit circle.
The terminal weight matrix $P$ can be solved from Lyapunov equation \eqref{Lyapunov eqn}.
%, and it is obtained by
%\begin{align*}
% P = \begin{bmatrix}
%		6.7 &22.2\\ 22.2 &106.8
%	\end{bmatrix}.
%\end{align*}
%It follows from \eqref{calculation of terminal set} and \eqref{quadratic prog of terminal set} that the terminal constraint can be calculated by 
%\begin{align*}
%  \mathcal{X}_f=\{x\in \mathbb{R}^2 | x^TPx\leq 4.15\}.
%\end{align*}
The terminal constraint can be calculated by using \eqref{calculation of terminal set} and \eqref{quadratic prog of terminal set}.

%The simulation result is displayed by Fig.~\ref{fig: SI linear},
%where it can be seen that the origin of the closed-loop system is asymptotically stable, and the transient process is completed in 7 steps.
%Constraints are all satisfied.

In Fig.~\ref{fig: SI linear}, the simulation results indicate that the closed-loop system reaches the origin within $6$ steps, and maintains at the origin for all future times, while states and input never violate their constraints. Comparatively, results of conventional MPC are also provided in Fig.~\ref{fig: SI linear}, where it can be seen that states and control converge asymptotically, but not in finite time.

\begin{figure}[!tbp]
	\begin{center}
		\includegraphics[scale=0.45]{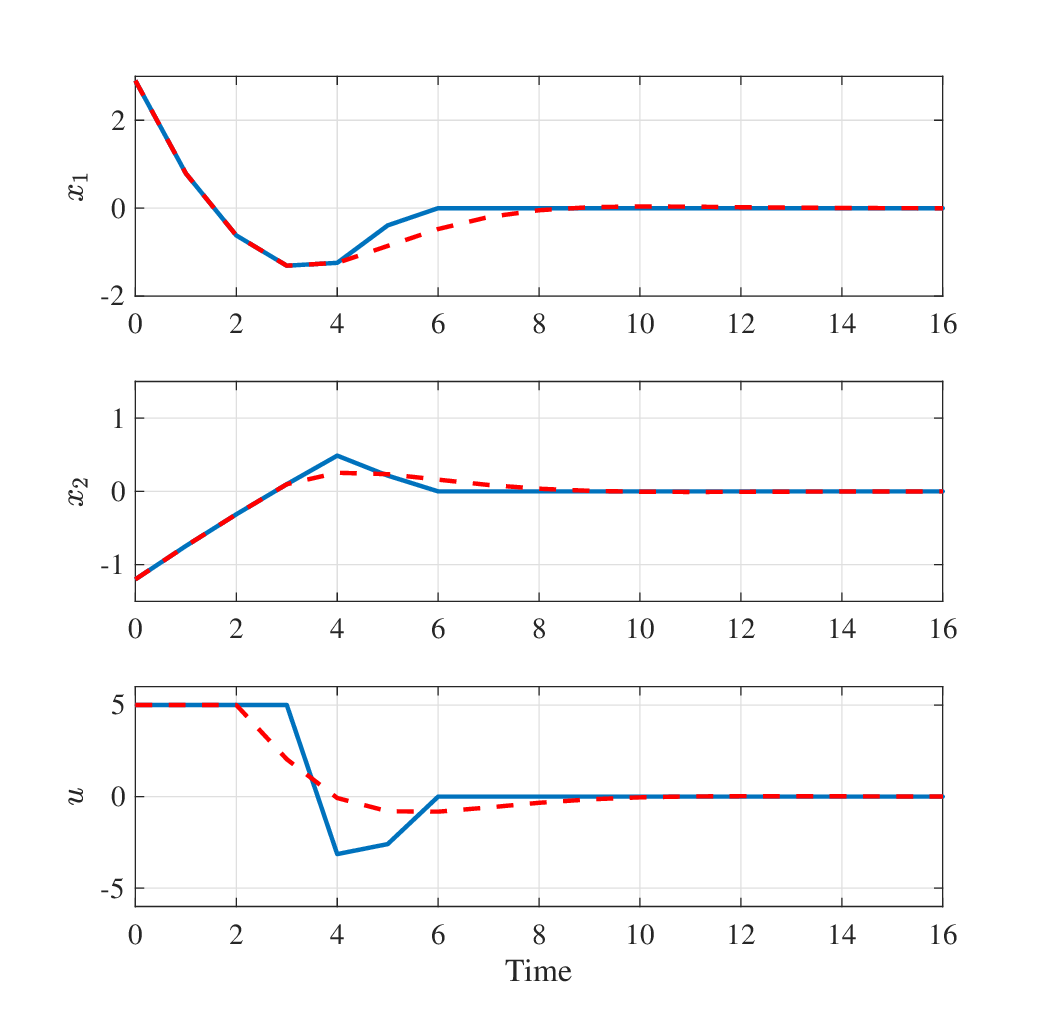}
		\caption{Single-input closed-loop states and control under the proposed infinite-horizon finite-time MPC (blue solid), with transient response settling in 6 steps. Comparatively, states and control with conventional MPC (Red dashed) converge asymptotically.} 
		\label{fig: SI linear}                             
	\end{center}                                
\end{figure}

Initial feasibility of the proposed infinite-horizon finite-time MPC (implemented by finite-horizon MPC with $N=8$) is compared with the previous terminal-cost strategy \cite{zhu2024finite}. The result is illustrated by Fig.~\ref{fig: feasibility},
where it can be witnessed that, with a prolonged control horizon, the proposed strategy is capable of enlarging the initial feasibility region significantly.
Both approaches achieve finite-time convergence.
%It has to be acknowledged that, with larger control horizon, the computational burden of the proposed finite-time MPC in this paper is greater than that in \cite{zhu2024finite}.  

\begin{figure}[!tbp]
	\begin{center}
		\includegraphics[scale=0.35]{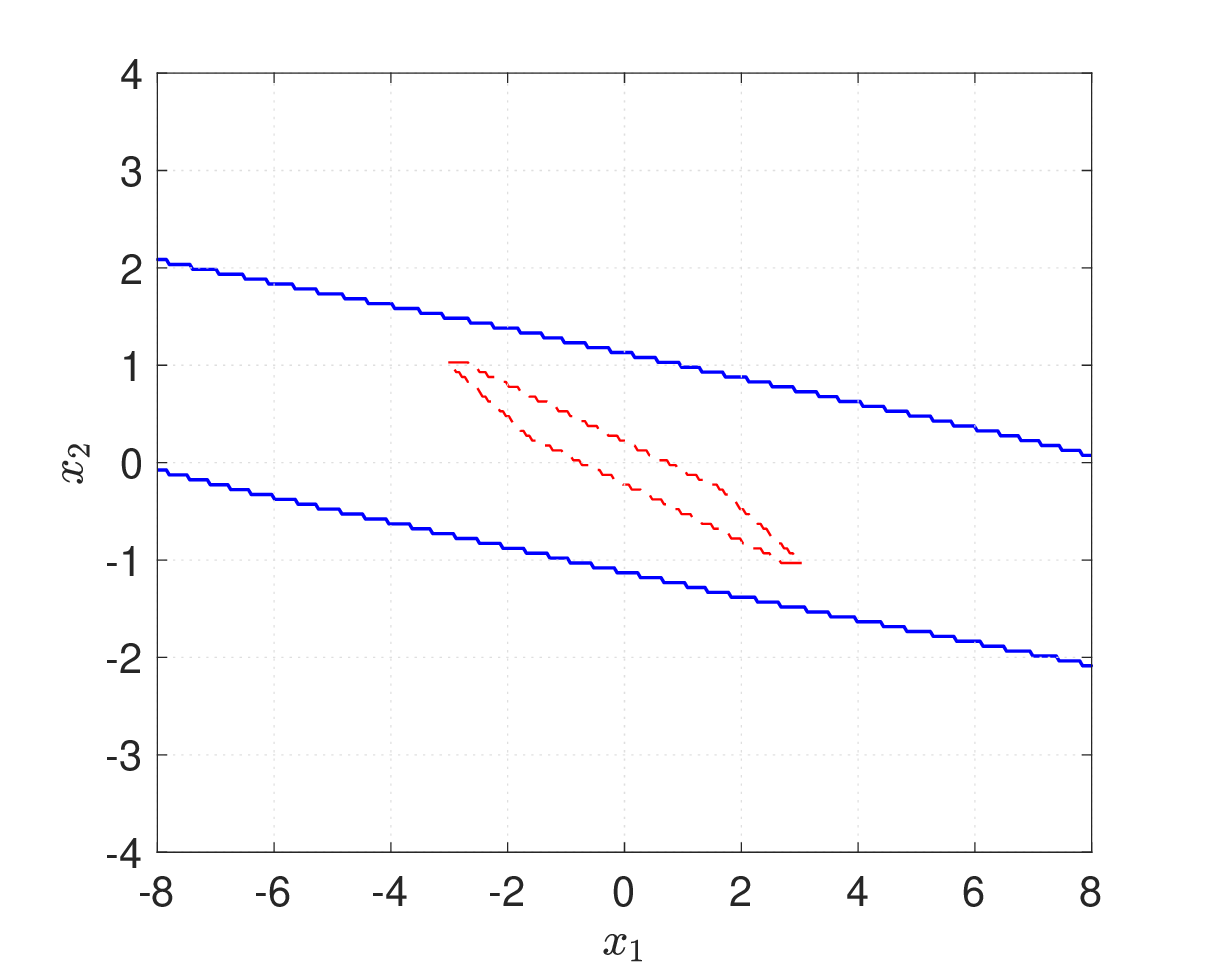}
		\caption{Comparison of initial feasibility regions: the previous terminal-cost strategy (red dashed); the proposed infinite-horizon strategy implemented via finite horizon $N=8$ (blue solid).} 
		\label{fig: feasibility}                             
	\end{center}                                
\end{figure}

%-------------------------------------------------------------------
\subsection{Finite-time MPC for multi-input linear systems}\label{subsec sim multi}

The multi-input linear plant is given by
\begin{align*}
	A &= \begin{bmatrix}
		1.47    &1.44   &-0.06   & 0.68    &0.08\\
		-0.17    &1.44    &1.06   &-0.43    &0.50\\
		-0.42    &1.04    &1.33   &-0.91    &0.51\\
		-0.99   & 0.12  & -0.02  & -0.77   & 0.39\\
		1.18   &-3.38   &-0.41    &2.62   &-0.41
	\end{bmatrix},\\
	B &= [b_1,~b_2],\\
	b_1 &= \begin{bmatrix}
		-0.0075 &0.085 &0.116 &0.023 &-0.048
	\end{bmatrix}^T,\\
	b_2 & =\begin{bmatrix}
		0.29 &-0.27 &-0.61 &-0.51 &1.91
	\end{bmatrix}^T,
\end{align*}
subject to constraints $|u_1|\leq 4$ and $|u_2|\leq 3$.

\begin{figure}
	\begin{center}
		\includegraphics[scale=0.45]{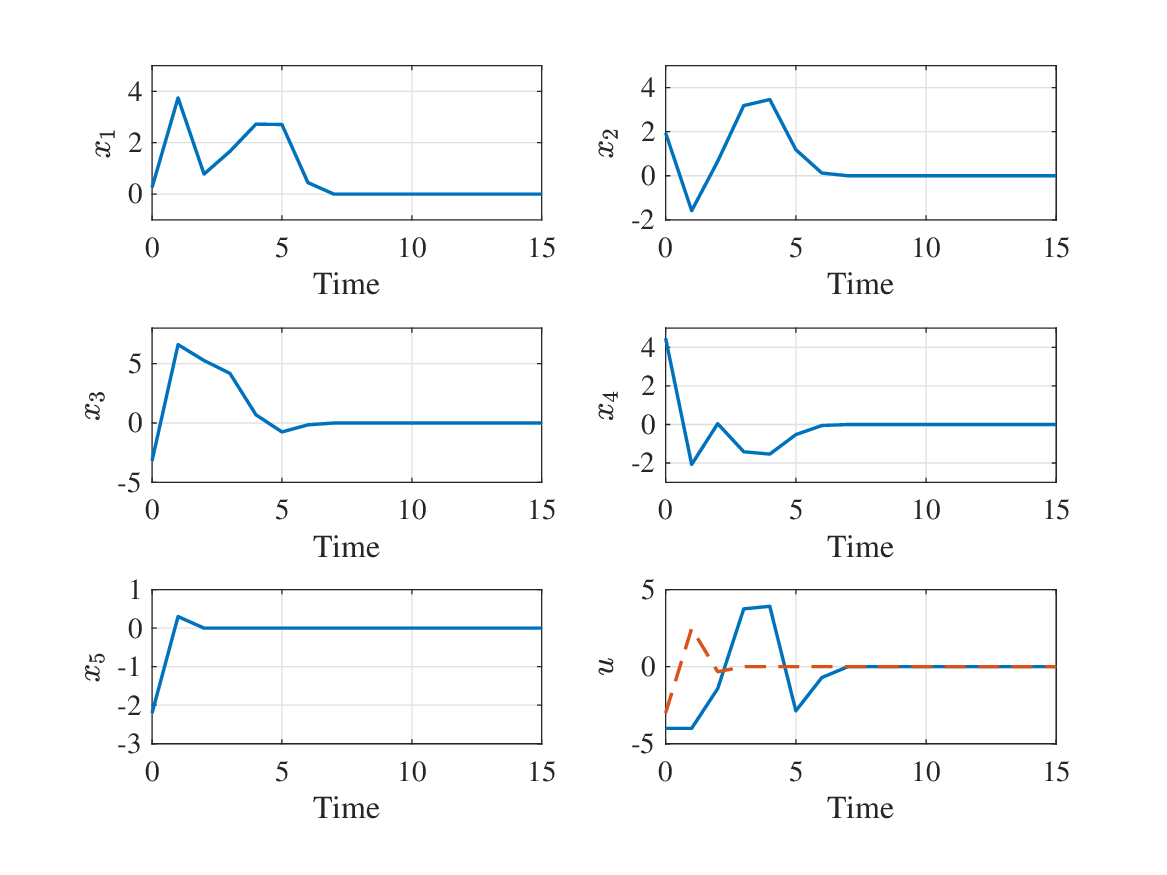}
		\caption{Multi-input closed-loop states and controls under the proposed infinite-horizon finite-time MPC, with transient response settling in 7 steps.} 
		\label{fig: MI linear}                             
	\end{center}                                
\end{figure}

The multi-input finite-time MPC is designed by following Theorem \ref{thm MI},
where the control horizon is selected by $N_1=N_2=10$.
Weighting matrices are selected by $Q_1=I_{3\times3}$, $Q_2=I_{2\times2}$, $R_1=0.1$ and $R_2=0.1$.

%The simulation result is shown in Fig.~\ref{fig: MI linear},
%where it can be seen the transient process completes in $11$ steps.
%It can be witnessed that $u_1$ is always within its bound, and $u_2$ remains zero. Zero value of $u_2$ is due to that $q=1$ in its controllable canonical form, and the system can be fully controlled by $u_1$.

The results, depicted in Fig.~\ref{fig: MI linear}, show that the transient process is completed in $7$ steps. The control inputs consistently satisfy their constraints.

%-----------------------------------------------------------------
\subsection{Finite-time MPC for nonlinear systems}

The nonlinear plant is given by
\begin{align*}
	x_1(k+1) =& -1.1x_1(k)+ 2\sin x_2(k),\\
	x_2(k+1) =& 0.2x_1(k)x_2(k) 1.05x_2(k) +0.79x_3(k),\\
    x_3(k+1) =& x_3(k) + u(k),
\end{align*}
subject to constraints $|x_2|< \pi/2$ and $|u|\leq 2$. 

The nonlinear finite-time MPC is designed by following Theorem \ref{thm nonlinear},
where the control horizon is set to $N=8$.
Weighting matrices are selected by $Q=I_{2\times2}$ and $R=0.1$.
%Performances of the proposed nonlinear finite-time MPC are illustrated by Fig.~\ref{fig: nonlinear}, where it can be seen that the transient process completes in 5 steps.
The results of nonlinear finite-time MPC are shown in Fig.~\ref{fig: nonlinear}, where it is shown that the closed-loop system reaches the origin within $8$ steps.
The state and control are always within their constraints.
Results of conventional nonlinear MPC with terminal cost is provided for comparison, where parameters are all set the same as the finite-time MPC, except that the optimization horizon is from $1$ to $N$.
It can be witnessed that states and control with conventional MPC converge asymptotically, but not in finite time.

\begin{figure}
	\begin{center}
		\includegraphics[scale=0.44]{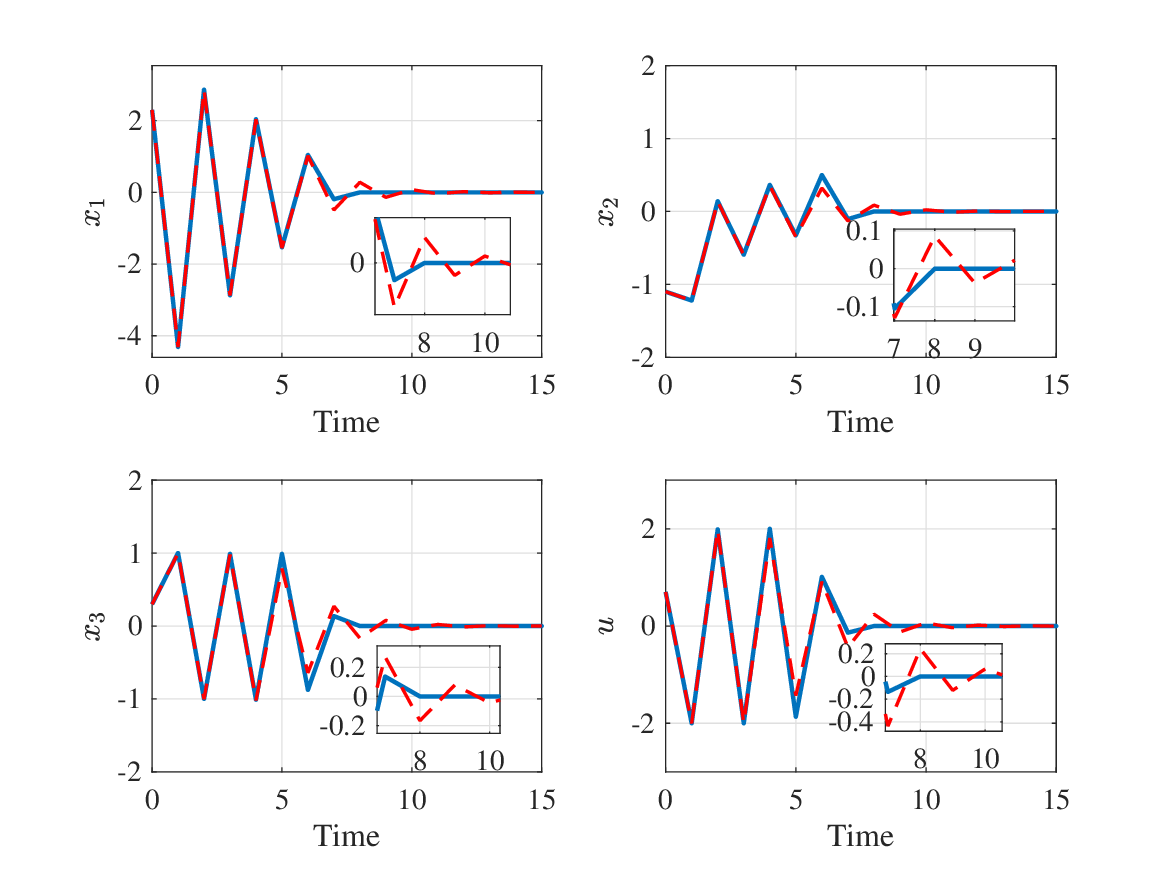}
		\caption{ Closed-loop states and control of the nonlinear system under the proposed finite-time MPC (blue solid): they converge to the origin in $8$ steps.
        Comparatively, states and control with conventional MPC (red dashed) converge asymptotically.} 
		\label{fig: nonlinear}                             
	\end{center}                                
\end{figure}

However, as a general characteristic of any finite-time control scheme, achieving strict finite-time convergence can often sacrifice transient performance. The focus is on driving the state to zero in a finite number of steps, which may lead to more aggressive control action and a sub-optimal path during the initial phase compared to conventional MPC approaches.

%-------------------------------

\subsection{Robustness to bounded disturbance}

%It is noted that the finite-time convergence is a special form of exponential or asymptotic stability. Due to the robustness of exponential or asymptotic stability, at least ultimate boundedness can be ensured in the presence of bounded disturbances. 

In this section, the example in Section \ref{subsec sim multi} is revisited in the presence of bounded random disturbance 
\begin{align*}
    x(k+1) = Ax(k) +B\left(u(k)\right)+w,
\end{align*}
where $A$ and $B$ are given in Section \ref{subsec sim multi}; the disturbance $w=[w_1,\cdots,w_5]^T$ is random at each time step $k$, and it satisfies $|w_i(k)|\leq 0.1$ for $i=1,\cdots,5$.
The simulation was run for 10 times, and the results are displayed in Fig.~\ref{fig: dist}. It can be witnessed that system states are ultimately bounded, and controls are always within their constraints.

Although the finite-time convergence is established for the nominal system, it converges asymptotically before entering the deadbeat region. Under a bounded additive disturbance, the closed-loop system therefore exhibits an ultimate boundedness property, where the size of ultimate bound depends on the disturbance magnitude and the closed-loop dynamics.

In the presence of small disturbances, recursive feasibility will not be lost, provided that the system state is not perturbed to leave its feasibility region. It is acknowledged that, if there exit excessively large modeling errors or disturbances, the recursive feasibility will possibly be lost.

\begin{figure}
	\begin{center}
		\includegraphics[scale=0.45]{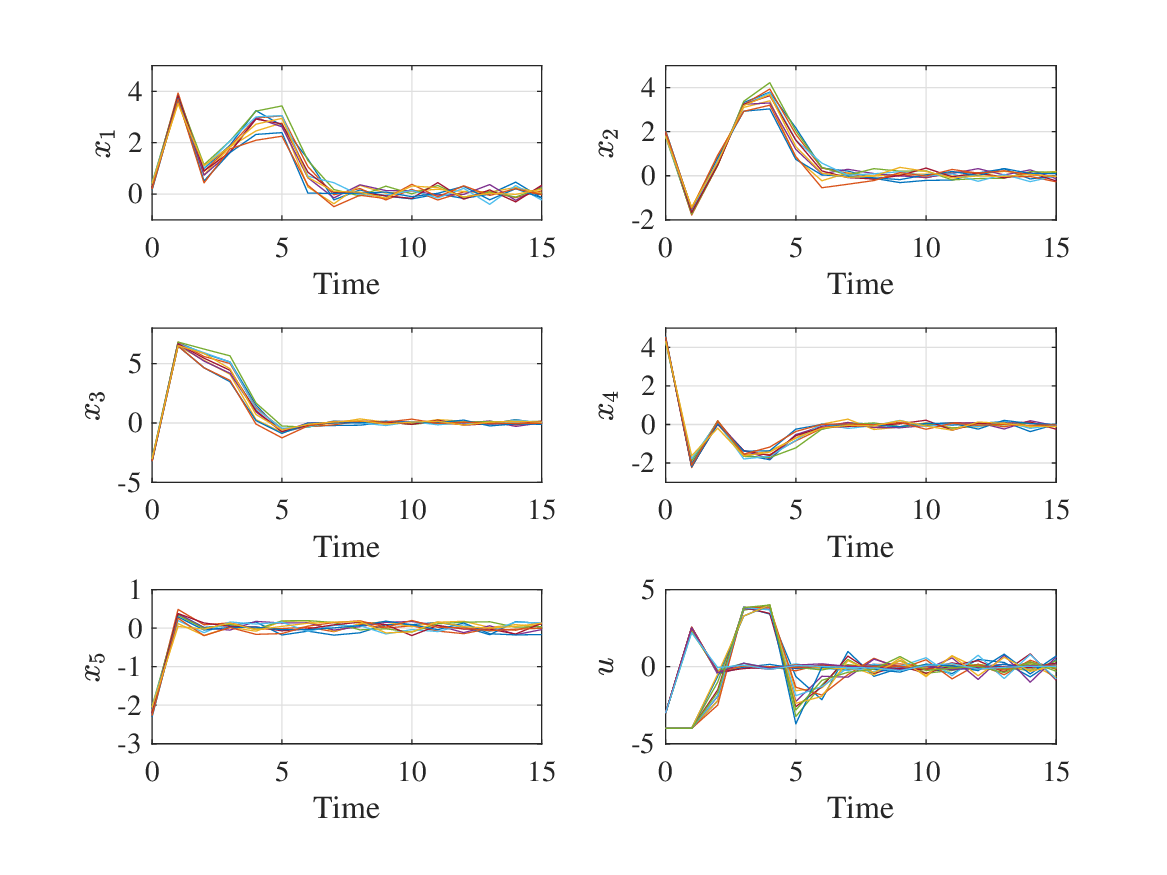}
		\caption{Closed-loop states and control in the presence of bounded random disturbance: system states are ultimately bounded.} 
		\label{fig: dist}                             
	\end{center}                                
\end{figure}

\section{Conclusion}

An infinite-horizon MPC framework is proposed to stabilize the constrained discrete-time system, such that the state of the closed-loop system reaches the origin in finite-time, and maintains thereafter. 
The proposed framework extends the previous terminal cost strategy to infinite-horizon stage cost, and it is implementable through an equivalent finite-horizon strategy.
Since the control horizon is prolonged significantly, the initial feasibility region can be enlarged.
It also avoids terminal equality constraint or switches inside the one-step reachability region.
The expected performances are proved theoretically, and are supported by simulation examples.
Results in this paper can be regarded as a general framework for constrained finite-time MPC design.

%The theoretical result proposed in this paper guaranties finite-time stabilization, but the explicit calculation of the finite convergence time $T$ depends on detailed constraints and design parameters. Developing methods to explicitly estimate or calculate a tighter upper bound for the finite convergence time is within our future works.

Assumption of feedback linearizability of nonlinear systems restricts the applicability of the proposed finite-time MPC.
It is within future works to further expand the applicability of the proposed finite-time MPC to nonlinear systems that are not feedback linearizable.

\balance

%\begin{thebibliography}{99}

%\bibliographystyle{apacite}
\bibliographystyle{IEEEtran}
\bibliography{deadbeat}

%\end{thebibliography}

\end{document}